\documentclass[journal,onecolumn]{IEEEtran}
\usepackage{amssymb}
\usepackage{amsmath}
\usepackage{longtable}
\newtheorem{theorem}{Theorem}[section]
\newtheorem{lemma}[theorem]{Lemma}

\newtheorem{example}{Example}

\newenvironment{proof}[1][Proof]{\begin{trivlist}
\item[\hskip \labelsep {\bfseries #1}]}{\end{trivlist}}
\newenvironment{definition}[1][Definition]{\begin{trivlist}
\item[\hskip \labelsep {\bfseries #1}]}{\end{trivlist}}

\newenvironment{remark}[1]
[Remark]{\begin{trivlist}
\item[\hskip \labelsep {\bfseries #1}]}{\end{trivlist}}

\makeatletter

\newcommand{\Rmnum}[1]{\expandafter\@slowromancap\romannumeral #1@}
\makeatother

\ifCLASSINFOpdf

\else

\fi

\begin{document}
\title{Two infinite classes of rotation symmetric bent functions with simple
representation}

\author{Chunming~Tang,
 Yanfeng~Qi, Zhengchun Zhou, Cuiling Fan% <-this % stops a space
\thanks{C. Tang is with School of Mathematics and Information, China West Normal University, Sichuan Nanchong, 637002, China. e-mail: tangchunmingmath@163.com
}

\thanks{Y. Qi is with School of Science, Hangzhou Dianzi University, Hangzhou, Zhejiang, 310018, China.
e-mail: qiyanfeng07@163.com.
}
\thanks{
Z. Zhou is with the School of Mathematics, Southwest Jiaotong University, Chengdu, 610031, China. e-mail:
zzc@swjtu.edu.cn.}
\thanks{
C. Fan is with the School of Mathematics, Southwest Jiaotong University, Chengdu, 610031, China. e-mail:
fcl@swjtu.edu.cn.}
}

% <-this % stops a space

\maketitle

\begin{abstract}
In the literature, few $n$-variable rotation symmetric bent functions have been constructed. In this paper, we present two infinite classes of rotation symmetric bent functions
on $\mathbb{F}_2^{n}$  of the two forms:

{\rm (i)} $f(x)=\sum_{i=0}^{m-1}x_ix_{i+m}
+ \gamma(x_0+x_m,\cdots,
x_{m-1}+x_{2m-1})$,

{\rm (ii)} $f_t(x)=
\sum_{i=0}^{n-1}(x_ix_{i+t}x_{i+m}
+x_{i}x_{i+t})+ \sum_{i=0}^{m-1}x_ix_{i+m}+
\gamma(x_0+x_m,\cdots,
x_{m-1}+x_{2m-1})$,

\noindent where $n=2m$,  $\gamma(X_0,X_1,\cdots,
X_{m-1})$ is any rotation symmetric polynomial, and $m/gcd(m,t)$ is odd.
The  class (i) of rotation symmetric  bent functions
has algebraic degree ranging from 2 to
$m$ and the other  class (ii)  has algebraic degree
ranging from 3 to $m$.
\end{abstract}


% \fi
\IEEEpeerreviewmaketitle

\section{Introduction}

Boolean bent functions introduced by Rothaus \cite{R1976} in 1976 are an interesting combinatorial object with the maximum
Hamming distance to the set of all affine functions. Such functions have been extensively studied because of
their important applications in cryptograph (stream ciphers \cite{C2010}), sequences  \cite{OSW1982}, graph theory
\cite{PTF2010}, coding theory ( Reed-Muller codes \cite{CHLL1997}, two-weight and three-weight linear codes   \cite{CK1986,Ding2015}), and association schemes \cite{PTFL2011}.
A complete classification of bent functions is still elusive. Further, not only
their characterization, but also
their generation are challenging problems.
Much work on bent functions are devoted to
the construction of bent functions
\cite{CCK2008,C1994,C1996,C2010,CM2011,CG2008,
CK2008,D1974,DD2004,DLCCFG2006,G1968,
L2006,LK2006,LHTK2013,M1973,M2009,M2010,M2011-1,
M2011-2,M2014,MF2013,YG2006}.

Rotation symmetric Boolean functions, introduced by Pieprzyk and Qu \cite{PQ1999}, are invariant under circular translation of indices. Due to less space to be stored and
allowing faster computation of the Walsh transform, they are of great interest.
They can be obtained from
idempotents (and vice versa) \cite{F1999,FF1998}.
Characterizing and
constructing  rotation symmetric bent  functions are difficult and have theoretical and practical interest.
The dual of a rotation symmetric  bent function is also a rotation symmetric bent function.  In the literature,
few constructions of bent idempotents have been presented, which are restricted by
the number of variables and have algebraic degree no more than 4.
See more rotation symmetric bent functions
in \cite{CGL2014,CGL2014-1,DMS2009,GZLC2012,SM2008,SMC2004,
ST2015}.

Quadratic rotation symmetric bent functions have been characterized by Gao et al.
\cite{GZLC2012}. They
proved that
the quadratic function
$$
\sum_{i=1}^{m-1}c_i(
\sum_{j=0}^{n-1}x_jx_{i+j})
+c_m(\sum_{j=0}^{m-1}x_jx_{m+j})
$$
is rotation symmetric bent if and only if
the polynomial $\sum_{i=1}^{m-1}
c_i(X^i+X^{n-i})+c_mX^m$ is
coprime with $X^n+1$, where
$c_i\in \mathbb{F}_2$.
 Stanica
et al.  \cite{SM2008} conjectured that there are no
homogeneous rotation sysmetric
bent functions of algebraic degree
greater than 2. The construction
of rotation symmetric bent functions
of algebraic degree greater than
2 is an interesting problem \cite{C2014}.
Charnes et al. \cite{CRB2002} constructed homogeneous bent functions of algebraic degree 3 in 8, 10, and 12 variables by applying the machinery of invariant theory.
Up to now, there are few known constructions of rotation symmetric bent functions.
 Gao et al. \cite{GZLC2012}
constructed an infinite class of
cubic rotation symmetric bent functions of
the from
\begin{equation*}
f_t(x_0,x_1,\cdots,x_{n-1})
=\sum_{i=0}^{n-1}(x_ix_{i+t}x_{i+m}
+x_ix_{i+t})+\sum_{i=0}^{m-1}
x_ix_{i+m},
\end{equation*}
where $1\leq t\leq m-1$ and
$m/gcd(m,t)$ is odd.
Carlet et al. \cite{CGL2014} presented
$n$-variable cubic
rotation symmetric bent functions of
the form
\begin{equation*}
f(x_0,x_1,\cdots,x_{n-1})
=\sum_{i=0}^{n-1}x_ix_{i+r}x_{i+2r}+
\sum_{i=0}^{2r-1}x_ix_{i+2r}x_{i+4r}+
\sum_{i=0}^{m-1}x_ix_{i+m},
\end{equation*}
where $n=2m=6r$.
Carlet et al. \cite{CGL2014-1} proposed
an infinite class of quartic
 rotation
symmetric bent functions
from two known semi-bent rotation symmetric functions by the indirect sum.
Su and Tang \cite{ST2015} gave a class of
$n$-variable rotation symmetric bent functions of any possible algebraic degree ranging from $2$ to $n/2$ of the form
\begin{equation}\label{ST-E}
f(x)=\sum_{i=0}^{m-1}
x_ix_{i+m}+
\sum_{\delta\in A}
\sum_{{\beta'\boxplus \beta''\in \mathcal{O}_m(\delta)}
}\prod_{i=0}^{m-1}x_i^{\beta'}
x_{i+m}^{\beta''},
\end{equation}
where
\begin{itemize}
\item $\delta\in \mathbb{F}_2^m$.
\item $\mathcal{O}_n(\delta)$ is
the orbit of $\delta$ by cyclic shift.
\item $A$ is a subset of the representative elements of all the
    orbits $\mathcal{O}_m(\delta)$.
\item $\beta'=(\beta_0',\beta_1',
\cdots,\beta_{m-1}')$ and
$\beta''=(\beta_0'',\beta_1'',
\cdots,\beta_{m-1}'')$.
\item $\boxplus$ denotes the sum over $\mathbb{Z}$.
\end{itemize}
These functions contain functions
by Carlet et al. \cite{CGL2014}.

Motivated by the constructions of
Gao et al. \cite{GZLC2012} and Su et al. \cite{ST2015}, this paper constructs new rotation symmetric
bent functions from some known rotation
symmetric bent functions. We obtain
two infinite classes of
rotation symmetric bent functions
which are equivalent to functions in
the class of Maiorana-McFarland.
Let $\gamma(X_0,X_1,\cdots, X_{m-1})$
be a rotation symmetric polynomial in
$\mathbb{F}_2[X_0,X_1,\cdots, X_{m-1}]$,
i.e., $\gamma(X_0,X_1,\cdots, X_{m-1})
=\gamma(X_1,\cdots,X_{m-1} ,X_{0})$.
We obtain two classes of rotation symmetric
bent functions of the form
\begin{align*}
&f(x)=\sum_{i=0}^{m-1}x_ix_{i+m}
+\gamma(x_0+x_m,\cdots,
x_{m-1}+x_m),\\
&f_t(x)=\sum_{i=0}^{n-1}(x_ix_{i+t}x_{i+m}
+x_{i}x_{i+t})+ \sum_{i=0}^{m-1}x_ix_{i+m}+
\gamma(x_0+x_m,\cdots,
x_{m-1}+x_{2m-1}),
\end{align*}
where $1\leq t\leq m-1$ and  $m/gcd(m,t)$ is odd. In fact, these bent functions
belong to the Maiorana-McFarland class
of bent functions.

The rest of the paper is organized as follows: Section 2 introduces some basic
notations, Boolean functions, rotation symmetric bent functions. Section
3 presents the  constructed
rotation  symmetric bent functions. Section 4 proves main results on rotation symmetric
bent functions. Section 5 makes a conclusion.

\section{Preliminaries}

Let $\mathbb{F}_2^n$ denote the $n$-dimensional vector space over
the finite field $\mathbb{F}_2$.
An $n$-variable Boolean function
$f(x_0,x_1,\cdots,x_{n-1})$ is a mapping
from $\mathbb{F}_2^n$ to $\mathbb{F}_2$.
And $f(x_0,x_1,\cdots,x_{n-1})$  can be represented by a polynomial called its algebraic normal form (ANF):
\begin{equation}\label{anf}
f(x_{0},x_1 \cdots, x_{n-1})=
\sum_{u\in \mathbb{F}_2^n}
c_{u}(\prod_{i=0}^{n-1}x_{i}^{\beta_i}),
\end{equation}
where $u=(\beta_0,\beta_1,\cdots,\beta_{n-1})$
and $c_{u}\in \mathbb{F}_2$. The number of variables in the highest order
product term with nonzero coefficient is
called its algebraic degree.

For simplicity, we call polynomials in
$\mathbb{F}_2[x_0,x_1,\cdots,x_{n-1}]$ of the form  in Equation (\ref{anf})
the reduced polynomials.
Hence, an $n$-variable Boolean function
is identified as  a reduced polynomial in
$\mathbb{F}_2[x_0,x_1,\cdots,x_{n-1}]$.

\begin{definition}
A Boolean function $f$ over
$\mathbb{F}_2^n$ or a reduced polynomial
$f$ in $\mathbb{F}_2[x_0,x_1,\cdots,x_{n-1}]$
is called  rotation symmetric if
for each input $x=(x_0,x_1,\cdots,x_{n-1})
\in \mathbb{F}_2^n$, we have
$$
f(x_1,x_2,\cdots,x_{n-1},x_0)
=f(x_0,x_1,\cdots,x_{n-1}).
$$
\end{definition}

The Walsh transform of a Boolean function calculates the correlations between the function and linear Boolean functions.
And the Walsh transform of
$f$ over $\mathbb{F}_{2}^n$
is
$$
\mathcal{W}_f(b)=
\sum_{x\in \mathbb{F}_{2}^n}
(-1)^{f(x)+
\sum_{i=0}^{n-1}x_ib_i},
$$
where $b=(b_0,b_1,\cdots,b_{n-1})
\in \mathbb{F}_{2}^n$ and $x=
(x_0,x_1,\cdots,x_{n-1})$.
\begin{definition}
A Boolean function $f:\mathbb{F}_{2}^n
\longrightarrow \mathbb{F}_2$ is  a bent function if  $\mathcal{W}_f(b)
=\pm 2^{n/2}$ for any $b\in
\mathbb{F}_2^n$.
\end{definition}
A Boolean bent function only exists
for even $n$. The algebraic degree of a bent function is no more than
$m$ for $n=2m\geq 4$ and
the algebraic degree of a bent function
for $n=2$ is 2.

Let $\sigma$ be a permutation of
$\mathbb{F}_2^n$ such that
for any bent function $f$, $f\circ \sigma$
is also bent. Then $\sigma(x)=xA+b$,
where $A$ is an $n\times n$ nonsingular binary matrix over $\mathbb{F}_2$,
$xA$ is the  product of the row-vector
$x$ and $A$,  and
$b\in\mathbb{F}_2^n$. All these
permutations
form an automorphism of the set of bent functions.
 Two
functions $f(x)$ and $g(x)=f\circ\sigma(x)$
are called linearly  equivalent.
If $f(x)$ is bent and $L(x)$ is an affine function, then $f+L$ is  also  a bent function. Two
functions $f$ and $f\circ \sigma +L$ are called EA-equivalent.
 The completed version of a class is the set of all functions EA-equivalent to the functions in the class.

Maiorana and McFarland \cite{M1973}  introduced independently a class of bent functions by concatenating affine functions.
This class  is called the Maiorana-McFarland class $\mathcal{M}$ of  functions defined over $\mathbb{F}_2^m
\times \mathbb{F}_2^m$ of the form
\begin{equation}\label{MM-F}
f(a,y)=y\pi(a)+h(a),
\end{equation}
where $(a,y)\in \mathbb{F}_2^m
\times \mathbb{F}_2^m$,
$\pi(a)$ is any mapping from
$\mathbb{F}_2^m$ to $\mathbb{F}_2^m$, and
$h(a)$ is any Boolean function on
$\mathbb{F}_2^m$. Then
$f$ is bent if and only if $\pi$ is bijective.

\section{Two infinite classes of
rotation symmetric bent functions}

In this section, we only present  two infinite classes of rotation symmetric bent functions. The proofs of the main results
will be given in the next section.

\begin{theorem}\label{rs-1}
Let $n=2m$, $\gamma(X_0,X_1,\cdots,
X_{m-1})\in \mathbb{F}_2[X_0,X_1,\cdots,
X_{m-1}]$ be a reduced polynomial
of algebraic degree $d$. Then the function
\begin{equation*}
f(x)=\sum_{i=0}^{m-1}x_ix_{i+m}
+ \gamma(x_0+x_m,\cdots,
x_{m-1}+x_{2m-1})
\end{equation*}
is a bent function.
Further, if $\gamma(X_0,X_1,\cdots,
X_{m-1})$ is rotation symmetric
, then $f$ is a rotation symmetric bent function. And if $d\geq 2$, then $f$ has algebraic degree
$d$.
\end{theorem}
\begin{example}
Let $m=6$. Then the function
$$
f(x)=\sum_{i=0}^{5}x_ix_{i+6}
+ \prod_{i=0}^5(x_i+x_{i+6})
$$
is a rotation symmetric bent function of algebraic degree 6.
\end{example}

\begin{theorem}\label{rs-2}
Let $n=2m$, $t$ be an integer such that
$1\leq t\leq m-1$ and
$m/gcd(m,t)$ is odd, and $\gamma(X_0,X_1,\cdots,
X_{m-1})\in \mathbb{F}_2[X_0,X_1,\cdots,
X_{m-1}]$ be a reduced polynomial.
Then the function
\begin{equation*}
f_t(x)=
\sum_{i=0}^{n-1}(x_ix_{i+t}x_{i+m}
+x_{i}x_{i+t})+ \sum_{i=0}^{m-1}x_ix_{i+m}+
\gamma(x_0+x_m,\cdots,
x_{m-1}+x_{2m-1})
\end{equation*}
is a bent function.
Further, if $\gamma(X_0,X_1,\cdots,
X_{m-1})$ is rotation symmetric of algebraic degree $d\geq 3$, then $f$ is a rotation symmetric bent function of algebraic degree $d$.
\end{theorem}
\begin{example}
Let $m=6$ and $t=2$. Then the function
$$
f_2(x)=
\sum_{i=0}^{11}(x_ix_{i+2}x_{i+6}
+x_{i}x_{i+2})+ \sum_{i=0}^{5}x_ix_{i+6}
$$
is a rotation symmetric  bent function of algebraic degree 6.
\end{example}
\begin{lemma}\label{lem-d}
Let $g(x_0,x_1,\cdots,x_{n-1})$ be a Boolean function on $\mathbb{F}_2^n$ or a reduced polynomial in $\mathbb{F}_2
[x_0,x_1,\cdots,x_{n-1}]$ such that

{\rm (1)} for any $0\leq i\leq m-1$,
$g(x_0,\cdots,x_i,\cdots,x_{i+m},
\cdots,x_{n-1})
=g(x_0,\cdots,x_{i+m},
\cdots, x_i,\cdots,x_{n-1})$.

{\rm (2)} for any $0\leq i\leq m-1$,
$x_ix_{i+m}$ is not in the terms of
$g$.

{\rm (3)} $g$ is rotation symmetric.
\end{lemma}

\noindent Then there exists a rotation symmetric polynomial
$\gamma(X_0,X_1,\cdots,X_{m-1})
\in \mathbb{F}_2[X_0,X_1,\cdots,X_{m-1}]$
such that
$$
g(x_0,x_1,\cdots,x_{n-1})
=\gamma(x_0+x_m,x_1+x_{m+1},\cdots,
x_{m-1}+x_{2m-1}).
$$
\begin{proof}
If there exists $\gamma(X_0,X_1,\cdots,X_{m-1})$
such that
$$
g(x_0,x_1,\cdots,x_{n-1})
=\gamma(x_0+x_m,x_1+x_{m+1},\cdots,
x_{m-1}+x_{2m-1}).
$$
Since $g$ is rotation symmetric, then
$\gamma(X_0,X_1,\cdots,X_{m-1})$ is rotation symmetric.

Now we will give the proof by the induction on algebraic degree $d$ of $g$, i.e,
there exits such rotation symmetric polynomial $\gamma$ from
rotation symmetric $g(x)$ of algebraic degree $d$ satisfying conditions (1) and
(2).

{\rm 1)} When $g=0$ or
$g=1$, such $\gamma$ obviously exits.

{\rm 2)} When $d=1$, such $\gamma$ obviously exits.

{\rm 3)} Suppose $d\geq 2$.
From the conditions (1) and (2), there exists $i$ such that
\begin{align*}
g(x_0,x_1,\cdots,x_{n-1})=&
x_ig'(x_0,\cdots,x_{i-1},x_{i+1},
\cdots,x_{i+m-1},x_{i+m+1},
\cdots,x_{n-1})\\
&+x_{i+m}g''(x_0,\cdots,x_{i-1},x_{i+1},
\cdots,x_{i+m-1},x_{i+m+1},
\cdots,x_{n-1}),
\end{align*}
where $g',g''\in \mathbb{F}_2(x_0,\cdots,x_{i-1},x_{i+1},
\cdots,x_{i+m-1},x_{i+m+1},
\cdots,x_{n-1})$. From the condition
(1), we have $g'=g''$.
From the induction of algebraic degree $d$,
for $g'$ and $g''$, there exists $\gamma'
(X_0,\cdots,X_{i-1},X_{i+1},\cdots,
X_{m-1})$ such that
$$
g'=g''=\gamma'(x_0+x_m,
\cdots,x_{i-1}+x_{i+m-1},
x_{i+1}+x_{i+m+1},\cdots,
x_{m-1}+x_{2m-1}).
$$
Take $\gamma(X_0,X_1,\cdots,X_{m})
=X_i\gamma'
(X_0,\cdots,X_{i-1},X_{i+1},\cdots,
X_{m-1})$. Then
$$
g(x_0,x_1,\cdots,x_{n-1})
=\gamma(x_0+x_m,x_1+x_{m+1},\cdots,
x_{m-1}+x_{2m-1}).
$$

Hence, this lemma follows.
\end{proof}
\begin{remark}
Let $g(x)=\sum_{i=0}^{m-1}
x_ix_{i+m}+
\sum_{\delta\in A}
\sum_{\beta'\boxplus\beta''\in \mathcal{O}_m(\delta)
}\prod_{i=0}^{m-1}x_i^{\beta'}
x_{i+m}^{\beta''}$ defined in
Equation (\ref{ST-E}).
We can verify that $g(x)$ satisfies
all the three conditions in Lemma
\ref{lem-d}.
There exists $\gamma(X_0,X_1,\cdots,
X_{m-1})$ such that
$$
f(x)=\sum_{i=0}^{m-1}x_ix_{i+m}
+ \gamma(x_0+x_m,\cdots,
x_{m-1}+x_{2m-1})
$$
is bent. This shows that
rotation symmetric bent functions constructed by Su and Tang \cite{ST2015}
are contained in functions in Theorem \ref{rs-2}.
\end{remark}

\section{Proofs of main results}

In this section, we give the proofs of our main results on rotation symmetric bent functions.
\subsection{The proof of Theorem \ref{rs-1}}

For any function $\gamma$ on
$\mathbb{F}_2^m$, the function
$$
f_0(a,y)
=\sum_{i=0}^{m-1}y_ia_i+
\gamma(a_0,a_1,\cdots,a_{m-1})
$$
is a bent function on
$\mathbb{F}_2^m\times \mathbb{F}_2^m$
in the Maiorana-McFarland class $\mathcal{M}$ of  functions defined in Equation (\ref{MM-F}).
Take  the nondegenerate linear transform
on $f_0(a,y)$  as
\begin{align*}
&y_i=x_i,\\
&a_i=x_i+x_{i+m},
\end{align*}
where $0\leq i\leq m-1$. We have a bent function
\begin{align*}
f_1(x_0,x_1,\cdots,x_{n-1})
=& \sum_{i=0}^{m-1}x_i(x_i+x_{i+m})
+\gamma(x_0+x_m,\cdots,x_{m-1}+x_{2m-1})\\
=& \sum_{i=0}^{m-1}x_ix_{i+m}
+\gamma(x_0+x_m,\cdots,x_{m-1}+x_{2m-1})
+  \sum_{i=0}^{m-1}x_i.
\end{align*}
Since $ \sum_{i=0}^{m-1}x_i$ is a linear
function, then $f(x)=
f_1+ \sum_{i=0}^{m-1}x_i$ is a bent function.
Further, if $\gamma$ is a rotation symmetric polynomial in $\mathbb{F}_2[X_0,X_1,
\cdots,X_{m-1}]$, then $f(x)$ is also rotation symmetric.

If $\gamma(X_0,X_1,\cdots,X_{m-1}) $
has algebraic degree $d$, then
$\gamma(x_0+x_m,\cdots,x_{m-1}+x_{2m-1})$
has algebraic degree $d$. If
$d\geq 3$, then the algebraic degree of
$f$ is $d$. Otherwise, $f$ has algebraic
degree less than 2. Thus,
$f$ has algebraic degree 2 since $f$ is bent.
Hence, Theorem \ref{rs-1} follows.

\subsection{The proof of Theorem \ref{rs-2}}

We start with the following lemma for the proof of Theorem \ref{rs-2}.

\begin{lemma}\label{lem-2}
Let $n=2m$, $1\leq t\leq m-1$,  and $x_0,x_1,\cdots,x_{n-1}
\in \mathbb{F}_2$. Then

{\rm (1)} $\sum_{i=0}^{m-1}
(x_ix_{i+t}+x_{i+m-t}x_{i+m})
=\sum_{i=0}^{n-1}x_ix_{i+t}
+\sum_{m-t}^{m-1}(x_ix_{i+t}
+x_{i+m}x_{i+m+t})$.

{\rm (2)} $\sum_{i=0}^{m-1}
(x_ix_{i+m+t}+x_{i-t}x_{i+m})
=\sum_{i=m-t}^{m-1}(x_ix_{i+m+t}
+x_{i+m}x_{i+t})$.

{\rm (3)}
$\sum_{i=0}^{m-1}
(x_ix_{i+t}+x_{i+m-t}x_{i+m}
+x_ix_{i+m+t}+
+x_{i-t}x_{i+m})
=\sum_{i=0}^{n-1}x_ix_{i+t}
+\sum_{i=m-t}^{m-1}(x_i+
x_{i+m}+1)(x_{i+t}
+x_{i+m+t}+1)
+
\sum_{i=m-t}^{m-1}(x_i+x_{i+t}
+x_{i+m}+x_{i+m+t}+1)$.

{\rm (4)} $
\sum_{i=0}^{m-1}x_ix_{i+m}
(x_{i+t}+x_{i+m+t})=
\sum_{i=0}^{n-1}x_ix_{i+t}x_{i+m}$.

{\rm (5)}
Let $0\leq i \leq m-1$,
$y_i=x_{i+m}+1$, and
$a_i=x_i+x_{i+m}+1$. Then
$(a_ia_{i+t}+a_{i+t}
+a_{i+m-t})y_i=x_ix_{i+m}(x_{i+t}
+x_{i+m+t})
+x_ix_{i+m}+(x_ix_{i+t}+x_{i+m-t}
x_{i+m})+(x_ix_{i+m+t}
+x_{i+m}x_{i-t})
+x_i+x_{i+m}+x_{i+m-t}
+x_{i-t}+1$.
\end{lemma}
\begin{proof}
{\rm (1)}
\begin{align*}
\sum_{i=0}^{m-1}
(x_ix_{i+t}+x_{i+m-t}x_{i+m})
=&\sum_{i=0}^{n-1}x_ix_{i+t}
+\sum_{i=m}^{2m-1}x_ix_{i+t}
+\sum_{i=0}^{m-1}x_{i+m-t}x_{i+m}\\
=&\sum_{i=0}^{n-1}x_ix_{i+t}
+\sum_{i=m}^{2m-1}x_ix_{i+t}
+\sum_{i=m-t}^{2m-t-1}x_{i}x_{i+t}\\
=&\sum_{i=0}^{n-1}x_ix_{i+t}
+\sum_{i=m-t}^{m-1}x_ix_{i+t}
+\sum_{i=2m-t}^{2m-1}x_{i}x_{i+t}\\
=&\sum_{i=0}^{n-1}x_ix_{i+t}
+\sum_{m-t}^{m-1}(x_ix_{i+t}
+x_{i+m}x_{i+m+t}).
\end{align*}

{\rm (2)}
\begin{align*}
\sum_{i=0}^{m-1}
(x_ix_{i+m+t}+x_{i-t}x_{i+m})
=& \sum_{i=0}^{m-1}
x_ix_{i+m+t}+\sum_{i=2m-t}^{m-1-t}
x_{i}x_{i+m+t}\\
=& \sum_{i=m-t}^{m-1}
x_ix_{i+m+t}+\sum_{i=2m-t}^{2m-1}
x_{i}x_{i+m+t}\\
=&
\sum_{i=m-t}^{m-1}(x_ix_{i+m+t}
+x_{i+m}x_{i+t}).
\end{align*}

{\rm (3)}
Let $S=\sum_{i=0}^{m-1}
(x_ix_{i+t}+x_{i+m-t}x_{i+m}
+x_ix_{i+m+t}
+x_{i-t}x_{i+m})$. From results (1)
and (2),
\begin{align*}
S=& \sum_{i=0}^{n-1}
x_ix_{i+t}
+\sum_{i=m-t}^{m-1}(x_ix_{i+t}+
x_{i+m}x_{i+m+t}
+x_ix_{i+m+t}
+x_{i+t}x_{i+m})\\
=& \sum_{i=0}^{n-1}
x_ix_{i+t}
+\sum_{i=m-t}^{m-1}(x_i+x_{i+m})
(x_{i+t}+x_{i+m+t})\\
=&  \sum_{i=0}^{n-1}x_ix_{i+t}
+\sum_{i=m-t}^{m-1}(x_i+
x_{i+m}+1)(x_{i+t}
+x_{i+m+t}+1)
+
\sum_{i=m-t}^{m-1}(x_i+x_{i+t}
+x_{i+m}+x_{i+m+t}+1).
\end{align*}

{\rm (4)}
\begin{align*}
\sum_{i=0}^{m-1}x_ix_{i+m}
(x_{i+t}+x_{i+m+t})=&
\sum_{i=0}^{m-1}x_ix_{i+t}x_{i+m}
+\sum_{i=0}^{m-1}x_{i+m}x_{i+m+t}x_i\\
=&\sum_{i=0}^{m-1}x_ix_{i+t}x_{i+m}
+\sum_{i=m}^{2m-1}x_ix_{i+t}x_{i+m}\\
=&\sum_{i=0}^{n-1}x_ix_{i+t}x_{i+m}
\end{align*}

{\rm (5)} Let $B=
(a_ia_{i+t}+a_{i+t}
+a_{i+m-t})y_i$. Then
\begin{align*}
B=&((a_i+1)a_{i+t}
+a_{i+m-t})y_i\\
=& (x_i+x_{i+m})(x_{i+t}
+x_{i+m-t}+1)y_i+a_{m-t+i}y_i\\
=& (x_{i+m}+1)(x_i+x_{i+m})(x_{i+t}
+x_{i+m-t}+1)+a_{m-t+i}y_i\\
=& x_i(x_{i+m}+1)(x_{i+t}
+x_{i+m-t}+1)+a_{m-t+i}y_i\\
=&x_i(x_{i+m}+1)(x_{i+t}
+x_{i+m-t}+1)+
(x_{m-t+i}+x_{i-t}+1)(x_{i+m}+1).
\end{align*}
Hence, this result can be obtained directly. \end{proof}

Define a class of functions on
$\mathbb{F}_2^m\times \mathbb{F}_2^m$ of the form
\begin{equation*}
f_0(a,y)=\sum_{i=0}^{m-1}
\pi_i(a)y_i+\gamma(1+a)+h_0(a),
\end{equation*}
where $a,y\in \mathbb{F}_2^m$
, $\pi_i(a)=a_ia_{i+t}+a_{i+t}
+a_{i+m-t}$, $h_0(a)
=\sum_{i=m-t}^{m-1}a_ia_{i+t}$,
and $\gamma\in \mathbb{F}_2[X_0,X_1,
\cdots,X_{m-1}]$.
Since $m/gcd(m,t)$ is odd, then from
Gao et al. \cite{GZLC2012}[Proof in Theorem 1],  $(a_0,a_1,\cdots,
a_{m-1})
\mapsto (\pi_0(a),\pi_1(a),\cdots,
\pi_{m-1}(a)) $ is a permutation of
$\mathbb{F}_2^m$. Then
$f_0(a,y)$ is a bent function.  Take the
affine transform on $f_0(a,y)$ as
$$
y_i=x_{i+m}+1, a_i=x_{i+m}+1, ~
0\leq i \leq m-1.
$$
This affine transform is nondegenerate.
Hence, $f_1(x)
=f_0(x_0+x_m+1,\cdots,x_{m-1}+
x_{2m-1}+1,x_{m}+1,\cdots, x_{2m-1}+1)$
is also bent. From Lemma \ref{lem-2},
\begin{align*}
f_1(x)=& \sum_{i=0}^{m-1}(a_ia_{i+t}+a_{i+t}
+a_{i+m-t})y_i+h_0(a)
+\gamma(x_0+x_m,\cdots,x_{m-1}+x_{2m-1})\\
=& \sum_{i=0}^{n-1}(x_ix_{i+t}x_{i+m}
+x_ix_{i+t})+\sum_{i=0}^{m-1}
x_ix_{i+m}+
\gamma(x_0+x_m,\cdots,x_{m-1}+x_{2m-1})\\
&+\sum_{i=m-t}^{m-1}
(x_i+x_{i+m}+1)(x_{i+t}+x_{i+t+m}
+1)+h_0(a)\\&
+\sum_{i=m-t}^{m-1}
(x_i+x_{i+m}+x_{i+m+t}+1)
+\sum_{i=0}^{m-1}(x_i+x_{i+m}
+x_{i+m-t}+x_{i-t}+1)\\
=& \sum_{i=0}^{n-1}(x_ix_{i+t}x_{i+m}
+x_ix_{i+t})+\sum_{i=0}^{m-1}
x_ix_{i+m}+
\gamma(x_0+x_m,\cdots,x_{m-1}+x_{2m-1})
+L(x),
\end{align*}
where $L(x)=\sum_{i=m-t}^{m-1}
(x_i+x_{i+m}+x_{i+m+t}+1)
+\sum_{i=0}^{m-1}(x_i+x_{i+m}
+x_{i+m-t}+x_{i-t}+1)$ is an affine function. Hence,  we have
$$
f(x)=f_1(x)+L(x)
$$
is a bent function.
When $\gamma$ is rotation symmetric,
$f(x)$ is also rotations symmetric.
Obviously, if $\gamma$ has algebraic
degree $d\geq 3$, then
$f$ is also a function of algebraic degree $d$.
Hence, Theorem \ref{rs-2} follows.

\begin{remark}
From the proofs of Theorem \ref{rs-1} and
Theorem \ref{rs-2}, bent functions
in both theorems are in
the completed Maiorana-McFarland class of bent functions.
\end{remark}

\section{conclusion}
In this paper, we propose a systematic method for constructing $n$-variable rotation symmetric bent functions from
some functions in the Maiorana-McFarland
class. One class of rotation symmetric bent functions
has algebraic degree ranging from 2 to
$m$ and the other class has algebraic degree
ranging from 3 to $m$.

\section*{Acknowledgment}
We would like to thank Professor M. 
Rotteler for helpful suggestion.
This work was supported by
the National Natural Science Foundation of China
(Grant No. 11401480, No.10990011 \& No. 61272499).
Yanfeng Qi also acknowledges support from
KSY075614050 of Hangzhou Dianzi University.


\ifCLASSOPTIONcaptionsoff
  \newpage
\fi

\end{document}